\RequirePackage{scrlfile} % Provides AfterClass
\AfterClass{article}{     % Make sure varioref is loaded before hyperref
  \RequirePackage{varioref}
}
\documentclass{icmcta4}

\usepackage{mathptmx}       % selects Times Roman as basic font
\usepackage{helvet}         % selects Helvetica as sans-serif font
\usepackage{courier}        % selects Courier as typewriter font
\usepackage[english]{babel}  % better hyphenation
\usepackage{amssymb} % math symbols
\usepackage{amsmath}
\usepackage[bottom]{footmisc}% places footnotes at page bottom
\usepackage{graphicx}        % the LaTeX graphic package
\usepackage{color}           % colourise letters
\usepackage{xspace}          % for \xspace = possible space after commands
\usepackage{url}             % hyperlinks
\usepackage{cite}            % sort and compress simultaneous citations
\usepackage{multirow}        % for multirow and multicolumn commands
\usepackage{bm}
\usepackage{booktabs}
\usepackage[utf8]{inputenc}
\usepackage{nicefrac}
\usepackage{tabularx}

\makeindex             % used for the subject index
\usepackage{makeidx}         % allows index generation

\usepackage{cleveref}   %note: use nameinlink option if using hyperref
\renewcommand{\cref}[1]{\Cref{#1}}
\renewcommand{\vref}[1]{\Vref{#1}}

\let\endproofi\endproof
\def\endproof{\qed\endproofi}
\usepackage{IEEEtrantools}
\newcommand\ifnull[3]{%
  \ifx\null#1%
    #2%
  \else%
    #3%
  \fi}

\newcommand{\word}[1]{\textnormal{#1}}
\newcommand\half{\tfrac 1 2}
\newcommand\defeq{\triangleq}            % define operator
\newcommand\modop{\ \word{mod}\ }         % mod infix operator (less spacing than \mod)
\newcommand{\mo}{{-1}}                   % minus one as one token (for saving {}'s )
\newcommand{\T}[1]{^{(#1)}}              % Indexes on e.g. sub-results from a algorithm loop
\renewcommand\vec[1]{\bm{#1}}            % Redefine what a vector looks like
\newcommand{\rev}[2][\null]{%
\ifnull{#1}{\overline{#2}}{%else
\presuper{[#1]}{#2}\overline{#2}}}

\newcommand\F{\mathbb F\xspace} %F
\newcommand\NN{\mathbb N\xspace}

\newcommand{\Code}{\mathcal C}
\newcommand\Errs{\mathcal E} % error position set
\newcommand\errs{\epsilon} % number of errors

\begin{document}
\title*{Power Decoding of Reed--Solomon Codes Revisited}

\author{Johan S.~R.~Nielsen}
\institute{Johan S.~R.~Nielsen \at Ulm University, Institute of Communications Engineering,
                  Ulm, Germany, \email{jsrn@jsrn.dk}}
\maketitle

\abstract{
  Power decoding, or ``decoding by virtual interleaving'', of Reed--Solomon codes is a method for unique decoding beyond half the minimum distance.
  We give a new variant of the Power decoding scheme, building upon the key equation of Gao.
  We show various interesting properties such as behavioural equivalence to the classical scheme using syndromes, as well as a new bound on the failure probability when the powering degree is 3.
}


\section{Introduction}

Power decoding was originally developed by Schmidt, Sidorenko and Bossert for low-rate Reed--Solomon codes (RS) \cite{schmidt_syndrome_2010}, and is usually capable of decoding almost as many errors as the Sudan decoder \cite{sudan_decoding_1997} though it is a unique decoder.
If an answer is found, this is always the closest codeword, but in some cases the method will fail; in particular, this happens if two codewords are equally close to the received. 
With random errors this seems to happen exceedingly rarely, though a bound for the probability has only been shown for the simplest case of powering degree 2 \cite{schmidt_syndrome_2010,zeh_unambiguous_2012}.

The algorithm rests on the surprising fact that a received word coming from a low-rate RS code can be ``powered'' to give received words of higher-rate RS codes having the same error positions.
For each of these received words, one constructs a classical key equation by calculating the corresponding syndromes and solves them simultaneously for the same error locator polynomial.

Gao gave a variant of unique decoding up to half the minimum distance \cite{gao_new_2003}: in essence, his algorithm uses a different key equation and with this finds the information polynomial directly.
We here show how to easily derive a variant of Power decoding for Generalised RS (GRS) codes, Power Gao, where we obtain multiple of Gao's type of key equation, and we solve these simultaneously.

We then show that Power Gao is \emph{equivalent} to Power syndromes in the sense that they will either both fail or both succeed for a given received word.
Power Gao has some ``practical'' advantages, though: it extends Power decoding to the case of using 0 as an evaluation point (which Power syndromes does not support); and the information is obtained directly when solving the key equations, so finding roots of the error locator and Forney's formula is not necessary.

The main theoretical advantage is that Power Gao seems easier to analyse: in particular, we show two new properties of Power decoding: 1) that whether Power decoding fails or not depends only on the error and not on the sent codeword; and 2) a new bound on the failure probability when the powering degree is 3.

We briefly sketched Power Gao already in \cite{nielsen_generalised_2013}, but its behaviour was not well analysed and its relation to Power syndromes not examined.

In \cref{sec:keyeq} we derive the powered Gao key equations, and in \cref{sec:solving} we describe the complete algorithm and discuss computational complexity issues.
In \cref{sec:props} we show the behavioural equivalence to Power syndromes as well as the new properties on Power decoding.

\section{The Key Equations}
\label{sec:keyeq}

Consider some finite field $\F$.
The $[n,k,d]$ Generalised Reed-Solomon (GRS) code is the set
\[
\Code = \big\{ \big( \beta_1f(\alpha_1), \ldots, \beta_nf(\alpha_n) \big) \mid f \in \F[x] \land \deg f < k \big\}
\]
where $\alpha_1,\ldots,\alpha_n \in \F$ are distinct, and the $\beta_1,\ldots,\beta_n \in \F$ are non-zero (not necessarily distinct).
The $\alpha_i$ are called \emph{evaluation points} and the $\beta_i$ \emph{column multipliers}.
$\Code$ has minimum distance $d = n-k+1$ and the code is therefore MDS.

Consider now that some $\vec c = (c_1,\ldots, c_n)$ was sent, resulting from evaluating some $f \in \F[x]$, and that $\vec r = (\beta_1 r_1,\ldots, \beta_n r_n) = \vec c + (\beta_1 e_1, \ldots, \beta_n e_n)$ was the received word with (normalised) error $\vec e = (e_1,\ldots, e_n)$.
Let $\Errs = \{ i \mid e_i \neq 0 \}$ and $\errs = |\Errs|$.
In failure probability considerations, we consider the $|\F|$-ary symmetric channel.

Introduce $G \defeq \prod_{i=1}^n(x-\alpha_i)$, and for any integer $t \geq 1$, let $R\T t$ be the Lagrangian polynomial through the ``powered'' $\vec r$, i.e.~the minimal degree polynomial satisfying $R\T t(\alpha_i) = r_i^t$ for $i=1,\ldots,n$.
Naturally, we have $\deg R\T t \leq n-1$ and $R\T t$ can be directly calculated by the receiver.
As usual for key equation decoders, the algorithm will revolve around the notion of error locator: $\Lambda = \prod_{j\in\Errs}(x-\alpha_j)$.
Choose now some $\ell \in \NN$ subject to $\ell(k-1) < n$.
Then we easily derive the powered Gao key equations:
\begin{proposition}
  \label{power_gao}
  $\Lambda R\T t \equiv \Lambda f^t \mod G$
\end{proposition}
\begin{proof}
  Polynomials are equivalent modulo $G$ if and only if they have the same evaluation at $\alpha_1,\ldots,\alpha_n$.
  For $\alpha_i$ where $e_i \neq 0$, both sides of the above evaluate to zero, while for the remaining $\alpha_i$ they give $\Lambda(\alpha_i){r_i}^t = \Lambda(\alpha_i)f(\alpha_i)^t$.
\end{proof}

\section{The Decoding Algorithm}
\label{sec:solving}

The key equations of \cref{power_gao} are non-linear in $\Lambda$ and $f$, so the approach for solving them is to relax the equations into a linear system, similarly to classical key equation decoding.
We will ignore the structure of the right hand-sides and therefore seek polynomials $\lambda$ and $\psi\T 1,\ldots,\psi\T \ell$ such that $\lambda R\T t \equiv \psi\T t \mod G$ as well as $\deg \lambda + t(k-1) \geq \deg \psi\T t$ for $t=1,\ldots,\ell$.
We will call such $(\lambda,\psi\T 1,\ldots, \psi\T \ell)$ \emph{a solution} to the key equations.

Clearly $(\Lambda, \Lambda f, \ldots, \Lambda f^\ell)$ is a solution.
There are, however, infinitely many more, so the strategy is to find a solution such that $\deg \lambda$ is minimal; we will call this the \emph{minimal solution}.
Thus decoding can only succeed when $\Lambda$ has minimal degree of all solutions.
The probability of this occurring will be discussed in \cref{sec:props}.

Conceptually, Power Gao decoding is then straightforward: pre-calculate $G$ and from the received word, calculate $R\T 1, \ldots, R\T \ell$.
Find then a minimal solution $(\lambda,\psi_1,\ldots,\psi_\ell)$ with $\lambda$ monic.
If this has the valid structure of $(\Lambda, \Lambda f, \ldots, \Lambda f^\ell)$, then return $f$.
Otherwise, declare decoding failure.

For Power syndromes, the key equations are similar to ours except that the modulo polynomials are just powers of $x$.
In this case, finding a minimal solution is known as multi-sequence shift-register synthesis, and the fastest known algorithm is an extension of the Berlekamp--Massey algorithm \cite{schmidt_syndrome_2010} or the Divide-\&-Conquer variant of this \cite{sidorenko_fast_2011}.
These can not handle the modulus $G$ that we need, however.

A generalised form of multi-sequence shift-register synthesis was considered in \cite{nielsen_generalised_2013}, and several algorithms for finding a minimal solution were presented.
The key equations for our case fit into this framework.
We refer the reader to \cite{nielsen_generalised_2013} for the details on these algorithms, but the asymptotic complexities when applied to Power Gao decoding are given in \vref{tab:compl}.
The same complexities would apply to Power syndromes and also match the algorithms \cite{schmidt_syndrome_2010,sidorenko_fast_2011} mentioned before.
The other steps of the decoding are easily seen to be cheaper than this; e.g.~the calculation of $R\T 1, \ldots, R\T \ell$ by Lagrangian interpolation can be done trivially in $O(\ell n^2)$ or using fast Fourier techniques in $O(\ell n \log^2 n)$ \cite[p. 231]{von_zur_gathen_modern_2012}.
Thus Power Gao decoding is asymptotically as fast as Power syndromes.

\section{Properties of the Algorithm}
\label{sec:props}

Power Gao will fail if $(\Lambda, \Lambda f, \ldots, \Lambda f^\ell)$ is not the found minimal solution, so the question is when one can expect this to occur.
Since the algorithm returns at most one codeword, it \emph{must} fail for some received words whenever $\errs \geq d/2$.
Whenever an answer is found, however, this must correspond to a closest codeword: any closer codeword would have its own corresponding error locator and information polynomial, and these would yield a smaller solution to the key equations.

We first show that Power syndromes is behaviourally equivalent to Power Gao.
We will need to assume that the evaluation points $\alpha_i \neq 0$ for all $i$, which is a condition for Power syndromes decoding.
This implies $x \nmid G$.
We will use a ``coefficient reversal'' operator defined for any $p \in \F[x]$ as $\rev p = x^{\deg p}p(x^\mo)$.

In Power syndromes decoding, one considers $\vec r\T t = (\beta_1 r_1^t,\ldots,\beta_n r_n^t)$ for $t=1,\ldots,\ell$ as received words of GRS codes with parameters $[n, t(k-1)+1, n-t(k-1)]$, resulting from evaluating $f^t$; these ``virtual'' codes have the same evaluation points and column multipliers as $\Code$.
The $\vec r\T t$ will therefore have the same error positions as $\vec r$, so the same error locator applies.
For each $t$, we can calculate the syndrome $S\T t$ corresponding to $\vec r\T t$,  which can be written as
\[
%\textstyle
  S\T t = \Big(\sum_{i=1}^n \frac {r^t_i \zeta_i} {1 - x\alpha_i}  \modop x^{n-t(k-1)+1} \Big)
\]
where $\zeta_i = \prod_{j \neq i}(\alpha_i-\alpha_j)^\mo$; see e.g.~\cite[p. 185]{roth_introduction_2006}.
By insertion one sees that
\[
  \rev \Lambda S\T t \equiv \Omega\T t \mod x^{n-t(k-1)+1}  ,\quad t=1,\ldots, \ell
\]
where $\Omega\T t$ is a certain polynomial satisfying $\deg \Omega\T t < \deg \Lambda$.
Note that we are using $\Lambda$ reversed; indeed, one often defines error-locator as $\prod_{i \in \Errs}(1-x\alpha_i) = \rev \Lambda$ when considering the syndrome key equation.
The decoding algorithm follows simply from finding a minimal degree polynomial $\rev \lambda$ such that $\omega\T t = (\rev \lambda S\T t \modop x^{n-t(k-1)+1})$ satisfies $\deg \lambda > \deg \omega\T t$ for all $t$.
The decoding method fails if $\rev\lambda \neq \gamma \rev\Lambda, \forall\gamma \in \F$.
We now have:

\begin{table}[t]
  \centering
  \caption{Complexities of solving the key equations for the three approaches discussed in \cite{nielsen_generalised_2013}.
  }
  \label{tab:compl}
  % We want the table and note to both be centered, but by default the note will start writing from the centre of the table.
  % Thus, record the height of the table using an invisible box and -vspace this amount when writing the note.
  \newdimen\height
  \def\thetable{%
    \begin{tabular}{@{}l@{\hspace{1em}}l@{}}
      Algorithm  &   $O$-complexity \\ \midrule
      Mulders--Storjohann &  $\ell^2 n^2$ \\
      Alekhnovich & $\ell^3 n \log^2 n \log\log n$ \\
      Demand--Driven* & $\ell n^2[\log n \log\log n]$
    \end{tabular}
  }
  \setbox0=\vbox{\thetable}
  \height=\ht0 \advance\height by \dp0
  \begin{tabularx}{\textwidth}{@{}X@{\hspace{2em}}p{5cm}@{\quad}}
      \centering
      \thetable
  & 
    \footnotesize{
    \vspace*{-.5\height}
    *: If $\Code$ is cyclic, then $G = x^n - 1$ since the $\alpha_i$ form a multiplicative group, and in this case the log-factors in square brackets can be removed.
    }
  \end{tabularx}
\end{table}

\begin{proposition}
  \label{prop:power_syn_fail_gao}
  Decoding using Power Gao fails if and only if decoding using Power syndromes fails.
\end{proposition}
\begin{proof}
  Note first that $R\T t = \sum_{i=1}^n r_i^t \zeta_i \prod_{j \neq i}(x-\alpha_j)$.
  By insertion we get $S\T t \equiv \rev R\T t \rev G^\mo \mod x^{n-t(k-1)+1}$ (since $x \nmid G$).
  Power Gao fails if there is some $\lambda \in \F[x]$ which is not a constant times $\Lambda$ and such that $\deg \lambda \leq \deg \Lambda$ and $\psi\T t = (\lambda R\T t \modop G)$ has $\deg \psi\T t < \deg \lambda + t(k-1) + 1$ for each $t = 1,\ldots,\ell$.
  This means there must be some $\omega\T t$ with $\deg \omega\T t \leq \deg \lambda-1$ such that
  \begin{IEEEeqnarray*}{rCl+c}
    \lambda R\T t - \omega\T t G                     & = & \psi & \iff \\
    \rev \lambda \, \rev R\T t - \rev \omega\T t \rev G & = & \rev\psi\T t x^{\deg G + \deg \lambda - 1 - (\deg \lambda+t(k-1))} & \implies \\
    \rev \lambda \, \rev R\T t &\equiv& \rev \omega\T t \rev G  \mod x^{n-t(k-1)-1}
  \end{IEEEeqnarray*}
  Dividing by $\rev G$, we see that $\rev \lambda$ and the $\rev \omega\T t$ satisfy the congruences necessary to form a solution to the Power syndromes key equation, and they also satisfy the degree bounds.
  Showing the proposition in the other direction runs analogously.
\end{proof}
\begin{corollary}[Combining \cite{schmidt_syndrome_2010} and \cref{prop:power_syn_fail_gao}]
  \label{dec_radius}
  \def\ellm{{\hat\ell}}
  Power Gao decoding succeeds if $\errs < d/2$. Let
  \[
    \tau(\ell) = \tfrac \ell {\ell+1} n - \half\ell(k-1) - \tfrac \ell {\ell+1}
  \]
  Then decoding will fail with high probability if $\errs > \tau(\ellm)$, where $1 \leq\ellm \leq \ell$ is chosen to maximise $\tau(\ell)$.
  \footnote{%
      Decoding may succeed in certain degenerate cases, see \cite[Proposition 2.39]{nielsen_list_2013}.
      Failure is certain when using the method of \cite{schmidt_syndrome_2010} since what it considers ``solutions'' are subtly different than here.
    }
\end{corollary}

Between the above two bounds, Power decoding will sometimes succeed and sometimes fail.
Simulations indicate that failure occurs with quite small probability.
The only proven bound so far is for $\ell = 2$ where for exactly $\errs$ errors occurring, we have
$P_f(\errs) < (\nicefrac q {q-1})^\errs q^{3(\errs - \tau(2))}/(q-1)$, \cite{schmidt_syndrome_2010,zeh_unambiguous_2012}.

We will give a new bound for $P_f(\errs)$ when $\ell = 3$, but we will first show a property which allows a major simplification in all subsequent analyses.

\begin{proposition}
  \label{prop:power_gao_inv}
  Power Gao decoding fails for some received word $\vec r$ if and only if it fails for $\vec r + \hat{\vec c}$ where $\hat{\vec c}$ is any codeword.
\end{proposition}
\begin{proof}
  We will show that Power Gao decoding fails for $\vec r = \vec c + \vec e$ if and only if it fails for $\vec e$ as received word; since $\vec c$ was arbitrary, that implies the proposition.
  
  Let $R_e\T t$ be the power Lagrangians for $\vec e$ as received word, i.e.~$R_e\T t(\alpha_i) = e_i^t$ for each $i$ and $t$, and let $R_e = R_e\T 1$.
  Consider a solution to the corresponding key equations $(\lambda, \psi_1,\ldots,\psi_\ell)$; i.e.~$\lambda R_e\T t \equiv \psi_t \mod G$ and $\deg \lambda + t(k-1) + 1 > \deg \psi_t$.
  Let as usual $R\T t$ be the power Lagrangians for $\vec r$ as received word and $R = R\T 1$.
  Note now that $R\T t \equiv R^t \mod G$ since both sides of the congruence evaluate to the same at all $\alpha_i$; similarly $R_e\T t \equiv R_e^t \mod G$.
  Since $r_i = f(\alpha_i) + e_i$ linearity implies that $R = f + R_e$.
  Define $\psi_0 = \lambda$ and note that then also for $t=0$ we have $\deg \lambda + t(k-1) + 1 > \deg \psi_t$.
  We then have the chain of congruences modulo $G$:
  \begin{IEEEeqnarray*}{l}
    \lambda R\T t \equiv \lambda R^t \equiv \lambda (f + R_e)^t \equiv \textstyle \lambda \sum_{s=0}^t \tbinom t s f^s R_e^{t-s}
      \equiv \textstyle \sum_{s=0}^t \tbinom t s f^s \psi_{t-s} \mod G
  \end{IEEEeqnarray*}
  Each term in the last sum has degree $s\deg f + \deg \psi_{t-s} < s(k-1) + \deg \lambda + (t-s)(k-1) + 1 = \deg \lambda + t(k-1) + 1$, which means that
  \[
   \textstyle
   \Big(\lambda,\ \sum_{s=0}^1 \tbinom 1 s f^s \psi_{1-s},\ \ldots\ ,\ \sum_{s=0}^\ell \tbinom \ell s f^s \psi_{\ell-s} \Big)
  \]
  is a solution to the key equations with $\vec r$ as a received word.
  The same argument holds in the other direction, so any solution to one of the key equations induces a solution to the other with the same first component; obviously then, their minimal solutions must be in bijection, which directly implies that they either both fail or neither of them fail.
\end{proof}

For the new bound on the failure probability, we first need a technical lemma:
\begin{lemma}
  \label{lem:poly_triple}
  Let $U \in \F[x]$ of degree $N$, and let $K_1 < K_2 < K_3 < N$ be integers.
  Let $S = \{ (f_1,f_2,f_3) \mid f_1f_3 \equiv f_2^2 \mod U,\  f_2 \textrm{ monic },\  \forall t . \deg f_t < K_t\}$.
  Then
  \begin{IEEEeqnarray*}{rCl+l}
    |S| & \leq & 3^{K_2-1}q^{K_2}                            & \textrm{if } K_1 + K_3 - 2 < N \\
    |S| & \leq & 2^{K_1+K_3-2}q^{K_1+K_2+K_3 - N-2} & \textrm{if } K_1 + K_3 -2 \geq N
  \end{IEEEeqnarray*}
\end{lemma}
\begin{proof}
  If $K_1 + K_3 - 2 < N$, then $f_1 f_3 \equiv f_2^2 \mod U$ implies $f_1 f_3 = f_2^2$.
  We can choose a monic $f_2$ in $(q^{K_2}-1)/(q-1)$ ways.
  For each choice, then $f_2$ has at most $K_2-1$ prime factors, so the factors of $f_2^2$ 
  can be distributed among $f_1$ and $f_3$ in at most $3^{K_2-1}$ ways.
  Lastly, the leading coefficient of $f_1$ can be chosen in $q-1$ ways.
  
  If $K_1 + K_3 - 2 \geq N$, then for each choice of $f_2$, the product $f_1f_3$ can be among $\{ f_2^2 + gU \mid \deg g \leq K_1 + K_3 - 2 - N \}$.
  This yields at most $q^{K_1+K_2+K_3-N-2}/(q-1)$ candidates for $f_1f_2$; each of these has at most $K_1+K_3-2$ unique prime factors, which can then be distributed among $f_1$ and $f_2$ in at most $2^{K_1+K_3-2}$ ways.
  Again, the leading coefficient of $f_1$ leads to a factor $q-1$ more.
\end{proof}

\begin{proposition}
  \label{prop:power_syn_fail_prob}
  For $\ell=3$, the probability that Power decoding (Gao or Syndrome) fails when $\errs > d/2$ is at most 
  \def\third{\tfrac 1 3}
  \begin{IEEEeqnarray*}{r+l}
    (\nicefrac q {(q-1)})^\errs (\nicefrac 3 q)^{2\errs - (n-2k+1)} q^{3(\errs - \tau(2)) + k-1} & \textrm{ if } \errs < \tau(2) - \third k + 1 \\
    (\nicefrac q {(q-1)})^\errs 2^{2(2\errs-d) + 2(k-1)}q^{4(\errs - \tau(3))-2}  & \textrm{ if } \errs \geq \tau(2) - \third k + 1
  \end{IEEEeqnarray*}
\end{proposition}
\begin{proof}
  By \cref{prop:power_gao_inv}, we can assume that $\vec c = 0$, i.e.~that $\vec r = \vec e$.
  That means $R\T t(\alpha_i) = 0$ for $i \notin \Errs$, so we can  write $R\T t = E\T t \Upsilon$ for some $E\T t$ with $\deg E\T t< \errs$, where $\Upsilon = G/\Lambda$ is the ``truth-locator''.
  Power Gao decoding fails if and only if there exists $(\lambda, \psi_1,\psi_2, \psi_3)$ such that $\lambda \neq \Lambda$, $\deg \lambda \leq \deg \Lambda$, $\deg \lambda + t(k-1) + 1 > \deg \psi_t$ for $t=1,2,3$ as well as
  \begin{IEEEeqnarray*}{rCl+c+rCl}
    \lambda R\T t &\equiv& \psi_t \mod G &\iff &
    \lambda E\T t &\equiv& \hat\psi_t \mod \Lambda
  \end{IEEEeqnarray*}
  where $\hat\psi_t = \psi_t/\Upsilon$.
  Note that $\psi_t$ must be divisible by $\Upsilon$ since both the modulus and the left-hand side of the first congruence is.

  Denote by $E$ the unique polynomial with degree less than $\errs$ having $E(\alpha_i) = e_i$ for $i \in \Errs$.
  For any $i \in \Errs$ then $(\lambda E\T t)(\alpha_i) = \lambda(\alpha_i)\Upsilon(\alpha_i)^\mo e_i^t$, which means $\lambda E\T t \equiv \hat \lambda E^t \mod \Lambda$ for some polynomial $\hat \lambda$.
  
  After having chosen error positions, drawing error values uniformly at random is the same as drawing uniformly at random from possible $E$.
  So given the error positions, the probability that Power decoding will fail is $T_\Lambda/(q-1)^\errs$, where $T_\Lambda$ is the number of choices of $E$ such that there exist $\hat \lambda, \hat \psi_1, \hat\psi_2, \hat \psi_3$ having
  \[
    \hat \lambda E^t \equiv \hat \psi_t \mod \Lambda, \quad t=1,2,3
  \]
  as well as $\deg \hat \psi_t < \deg \Lambda + t(k-1) + 1 - (n-\deg \Lambda) = 2\errs - (n-t(k-1)-1)$.

  Note that these congruences imply $\hat\psi_1 \hat\psi_3 \equiv \hat\psi_2^2 \mod \Lambda$.
  Denote by $\hat T_\Lambda$ the number of triples $(\hat\psi_1, \hat\psi_2, \hat\psi_3) \in \F[x]^3$ satisfying just this congruence as well as the above degree bounds.
  Then $\hat T_\Lambda \geq T_\Lambda$: for if $\gcd(\hat \lambda, \Lambda) = 1$ then two different values of $E$ could not yield the same triple since $E \equiv \hat\psi_2/\hat\psi_1 \mod \Lambda$ uniquely determines $E$.
  Alternatively, if $\gcd(\hat \lambda, \Lambda) = g \neq 1$ then the congruences imply $g \mid \hat\psi_t$ for all $t$, so that $E \equiv (\hat\psi_2/g)/(\hat\psi_1/g) \mod \Lambda/g$.
  This leaves a potential $q^{\deg g}$ possible other choices of $E$ yielding the same triple; but all these possibilities are counted in the triples since $(t\psi_1/g, t\psi_2/g, t\psi_3/g)$ will be counted for any $t \in \F[x]$ with $\deg t < \deg g$.

  In fact, we have $\hat T_\Lambda \geq (q-1)T_\Lambda$, since whenever $(\hat\psi_1,\hat\psi_2,\hat\psi_3)$ is counted, so is $(\beta\hat\psi_1,\beta^2\hat\psi_2, \hat\psi_3)$, and this doesn't change the fraction $\hat\psi_1/\hat\psi_2$.
  Thus, we over-estimate instead $\hat T_\Lambda/(q-1)$ by counting the number of triples where $\hat\psi_2$ is monic.
  \cref{lem:poly_triple} gives an upper bound for exactly this number, setting $N=\errs$ and $K_t = 2\errs - (n-t(k-1)-1)$.
  Divided by $(q-1)^\errs$, this is then an upper bound on the failure probability given the error positions.
  But since this probability is independent of the choice of $\Lambda$, it is also the failure probability over all errors vectors of weight $\errs$.
\end{proof}

By experimentation, one can demonstrate that the bound can not be tight: for instance, for a $[250, 30, 221]$ GRS code, the bound is greater than 1 for $\errs > 143$, while simulation indicate almost flawless decoding up to $147$ errors.
However, in a relative and asymptotic sense the above bound is strong enough to show that up to $\tau(3)$ errors can be corrected with arbitrary low failure probability:

\begin{corollary}
  Having $\ell = 3$, then for any $\delta > 0$, with $n \rightarrow \infty$ while keeping $q/n$, $k/n$ and $\errs/n$ constant, the probability that Power decoding fails goes to 0 when $\errs/n < \tau(3)/n - \delta$.
\end{corollary}
\begin{proof}[Proof sketch]
  We consider only the high-error failure probability of \cref{prop:power_syn_fail_prob}.
  For $n \rightarrow \infty$, the failure probability bound will approach
  \begin{IEEEeqnarray*}{rCl}
    2^{2(2\errs-d) + 2(k-1)} q^{4(\errs - \tau(3))} &\leq& (q^n)^{4(\errs/n - \tau(3)/n) + (2(2\errs/n-d/n) + 2k/n)/\log q}
  \end{IEEEeqnarray*}
  The contribution $(2(2\errs/n-d/n) + 2k/n)/\log q$ goes to $0$ as $n \rightarrow \infty$, leaving $(q^n)^{-a}$ for $a = 4(\errs/n - \tau(3)/n) < -4\delta$.
\end{proof}

\bibliography{bibtex_abbr}

% trigger a \newpage just before the given reference
% number - used to balance the columns on the last page
% adjust value as needed - may need to be readjusted if
% the document is modified later
%\IEEEtriggeratref{8}
% The "triggered" command can be changed if desired:
%\IEEEtriggercmd{\enlargethispage{-5in}}

\end{document}